\documentclass[conference]{IEEEtran}
\IEEEoverridecommandlockouts
\usepackage{filecontents, subfig, float, subfloat, array, xspace, comment}
\usepackage{cite}
\usepackage{amsmath,amssymb,amsfonts}
\usepackage{algorithmic}
\usepackage{graphicx}
\usepackage{textcomp}
\usepackage{xcolor}
\usepackage[font=small]{caption}

\usepackage{tikz} 
\usepackage{adjustbox} 
\usepackage{pgfplots}

\usetikzlibrary{decorations.pathreplacing,calc,shapes,positioning,tikzmark} 

\usepackage[rightcaption]{sidecap}

\usepackage{listings}

\usepackage{amsthm}
\newtheorem{theorem}{Theorem}
\newtheorem{lemma}{Lemma}
\usepackage[linesnumbered,ruled,vlined]{algorithm2e}
\usepackage{amssymb}
\usepackage{algorithmic}
\SetKwInput{KwInput}{Input}                
\SetKwInput{KwOutput}{Output}              
\SetKwInput{KwThreada}{thread 1}
\SetKwInput{KwThreadb}{thread 2}

\def\BibTeX{{\rm B\kern-.05em{\sc i\kern-.025em b}\kern-.08em
    T\kern-.1667em\lower.7ex\hbox{E}\kern-.125emX}}
    
\newcommand{\scap}{\textsc{Escape}\xspace}

\usepackage[hyphens]{url}
\usepackage[hidelinks]{hyperref}
\hypersetup{breaklinks=true}
\begin{document}


\title{\scap to Precaution against Leader Failures}

\author{\IEEEauthorblockN{Gengrui Zhang}
\IEEEauthorblockA{\textit{Department of Electrical \& Computer Engineering}\\
\textit{University of Toronto}\\
Toronto, Canada \\
gengrui.zhang@mail.utorotno.ca}
\and
\IEEEauthorblockN{Hans-Arno Jacobsen}
\IEEEauthorblockA{\textit{Department of Electrical \& Computer Engineering}\\
\textit{University of Toronto}\\
Toronto, Canada \\
gengrui.zhang@mail.utorotno.ca}
}

\maketitle

\begin{abstract}
Leader-based consensus protocols must undergo a view-change phase to elect a new leader when the current leader fails. The new leader is often decided upon a candidate server that collects votes from a quorum of servers. However, voting-based election mechanisms intrinsically cause competition in leadership candidacy when each candidate collects only partial votes. This split-vote scenario can result in no leadership winner and prolong the undesired view-change period. In this paper, we investigate a case study of Raft's leader election mechanism and propose a new leader election protocol, called \scap, that fundamentally solves split votes by prioritizing servers based on their log responsiveness. \scap dynamically assigns servers with a configuration that offers different priorities through Raft's periodic heartbeat. In each assignment, \scap keeps track of server log responsiveness and assigns configurations that are inclined to win an election to more up-to-date servers, thereby preparing a pool of prioritized candidates. Consequently, when the next election takes place, the candidate with the highest priority will defeat its counterparts and becomes the next leader without competition. The evaluation results show that \scap progressively reduces the leader election time when the cluster scales up, and the improvement becomes more significant under message loss.
\end{abstract}


\section{Introduction}
Leader-based consensus algorithms have been widely deployed in practical systems, such as  HDFS~\cite{shvachko2010hadoop}, RAMCloud~\cite{ousterhout2010case}, Chubby~\cite{burrows2006chubby}, and ZooKeeper~\cite{hunt2010zookeeper}, and extensively studied in academia, such as Paxos~\cite{lamport1998part}, Viewstamped replication~\cite{oki1988viewstamped}, and Raft~\cite{ongaro2014search}. Under normal operation, they utilize a designated server as a leader to efficiently conduct consensus that satisfies state machine replication properties. However, leader-based algorithms are vulnerable to single points of failures. When the leader fails, the consensus process cannot proceed, and the system must select a new leader server through a view-change/leader election phase.

The view-change phase is undesired for providing high available services because no consensus can be reached without the coordination by the leader. Unfortunately, besides server crashes~\cite{sahoo2004failure, servercrash}, leader failures often take place because of various reasons, such as hardware failures~\cite{hardwarefails}, storage failures~\cite{bowers2011tell}, and, most commonly, network failures~\cite{gill2011understanding}. These issues may result in frequent leadership changes in systems that use timeouts to detect leader failures~\cite{ongaro2014search, hunt2010zookeeper, redisLeaderElection}. Before a new leader is elected, the system becomes temporarily unavailable and endures an out-of-service (OTS) period. For large-scale online applications, even seconds of OTS time is detrimental to user experiences and the quality of service~\cite{googleServicesOutages, amazonServicesOutages}. Therefore, the completion time of leader election significantly affects the performance of consensus services.

Voting-based mechanisms have been often chosen to implement leader election by selecting a server that can collect votes from a quorum (e.g., simple majority). A typical example is Raft's leader election mechanism~\cite{ongaro2014search}. Under normal operation, Raft uses a strong leadership for log replication where its leader dominates the consensus process. The strong leadership forces servers (operating as \emph{followers}) to only passively respond to the leader. Due to its modularity and simplicity, Raft quickly gained popularity and has been widely deployed~\cite{baiduRaft, sakic2017response, androulaki2018hyperledger, brown2016corda, zhang2021prosecutor}. However, Raft's leader election intrinsically creates competition among leadership candidates, which may result in split votes that needlessly prolong the leader election time. 

In Raft, each follower starts a timer when joining the system and resets the timer upon receiving a heartbeat from the leader. To maintain its leadership, the leader sends periodic heartbeats to followers. If a follower triggers a timeout, it transitions to a \emph{candidate} state and begins a leader election campaign to solicit votes. If the candidate successfully collects votes from a majority of servers, it becomes the new leader. However, Raft does not discriminate candidates as long as they are up-to-date. When multiple candidates start their election campaigns simultaneously and each of them collects only partial votes, no candidate can be elected by collecting votes from a majority because votes are split. The system has to repeat the election process until a candidate collects sufficient votes and becomes the new leader.

A simple solution to mitigate split votes is to use randomized timeouts; the randomness can reduce the probability of concurrent candidates that simultaneously initiate new leader election campaigns. Raft's evaluation shows the reduction of competing candidates when adding randomness to timeouts in a $5$-server cluster. However, this method does not fundamentally solve the problem of split votes and may become ineffective and inefficient when servers are not fully synchronized, especially at a large scale. 
To avoid concurrent election campaigns, the amount of randomness of election timeouts needs to be significantly increased when a system has more servers; the detection time of leader failures subsequently increases. In addition, stale servers (not fully synchronized) can interrupt election campaigns of up-to-date servers. When the timers of stale servers have a small initial timeout, they may often trigger timeouts before up-to-date servers, which may defer the appearance of a correct leader.



To address the above problem, in this paper, we propose a new leader election protocol, namely \scap, that takes precautions against leader failures by preparing a pool of ``future leaders'' before potential election campaigns take place. When the current leader fails, \scap is able to elect a new leader by one election campaign, escaping from leadership competition that causes split votes. 

\scap consists of two key components: the stochastic configuration assignment (SCA) and probing patrol function (PPF). Specifically, SCA assigns each server with a unique configuration that contains two paired parameters: a \emph{priority} and an \emph{election timeout}. The priority is an integer that determines the growth of a server's term (the logical time) while the election timeout defines the initial timeout of the server's timer. Initially, configurations are generated stochastically within a range (e.g., priorities are integers from [$1$, $n$] and election timeouts from [100, 200~ms]), where no two servers adopt the same configuration. 

In addition to SCA, \scap dynamically and atomically rearranges configurations among servers through the probing patrol function (PPF) and prepares a pool of prioritized candidates according to their log responsiveness. In principle, PPF arranges configurations that are inclined to win an election to servers that are up-to-date. To achieve this goal, first, PPF keeps track of each server's log index through the periodic leader-to-followers heartbeat. Then, it assigns higher-priority configurations to more up-to-date servers. Next, it distributes the new configurations in the following heartbeat. Finally, each server updates its priority and election timeout according to the received new configuration. Consequently, \scap prepares a pool of candidates as ``future leaders'' with differently prioritized configurations. If the current leader crashes, a candidate with a higher priority is able to defeat a competitor with a lower priority. A new leader will be elected without suffering from unnecessary leadership competitions.

Although there is no one-size-fits-all leader election mechanism for all consensus protocols, \scap can be adapted to support other election protocols, such as ZooKeeper~\cite{hunt2010zookeeper}, Redis cluster election~\cite{redisLeaderElection} and Azure leader election~\cite{azurele}. It also maintains Raft's understandability and the proposed improvement is simple to implement. By applying \scap, the system can always keep a pool of prioritized future leaders under normal operation, preventing potential conflicts before a new leader election takes place.

In this paper, we make the following contributions:
\begin{enumerate}
  \item Raft leader election analysis. We analyzed the split-vote scenarios that prolong Raft's leader election. The analysis addresses the tradeoff by adding randomness to election timeouts to mitigate split votes through our evaluation of a Raft cluster.

   \item The \scap leader election protocol. We designed two major components for \scap: stochastic configuration assignments and the probing patrol function. Configurations are assigned with priorities and rearranged by the probing patrol function, solving candidacy competition by periodically prioritizing servers.

  \item Experimental comparisons of leader election time. We implemented \scap and Raft and evaluated their performance of leader election under various fault scenarios at various scales. We distill the advantages of \scap based on experimental results.
      
\end{enumerate}

The remainder of the paper is organized as follows: Section~\ref{sec:background} provides Raft basics;
Section~\ref{sec:timer-randomization-tradeoff} states the problem of split votes and timer randomization;
Section~\ref{sec:dple} introduces the \scap design in detail; Section~\ref{sec:correctness} discusses correctness properties for \scap; Section~\ref{sec:evaluation} reports on the experimental evaluation of \scap\ against Raft; and Section~\ref{sec:related-work} presents related work.

\section{Background}
\label{sec:background}

\begin{SCfigure}[0.9][t]
\centering
\includegraphics[width=0.65\linewidth]{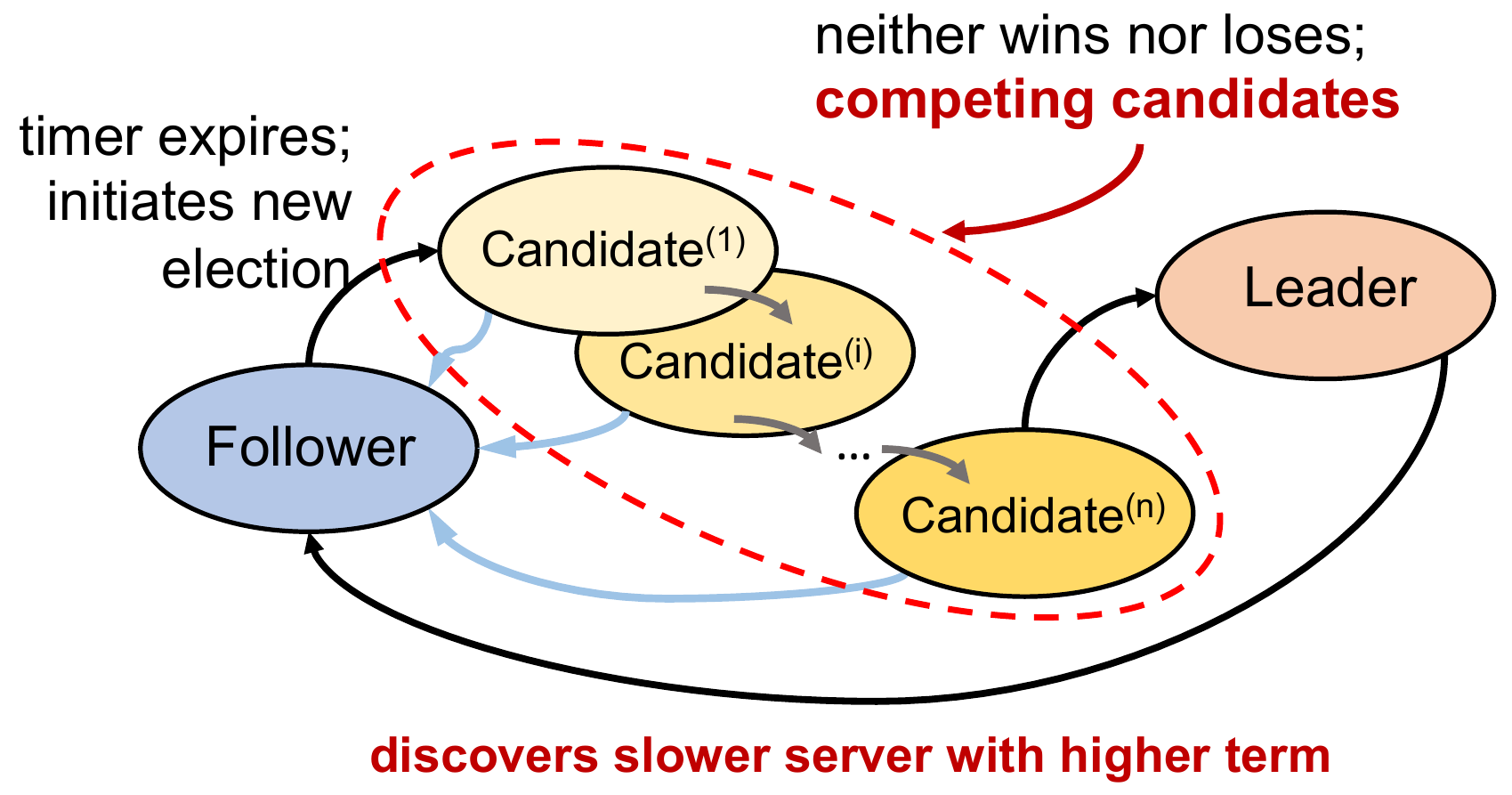}
\hfil
\caption{Server state transitions in Raft.
Due to split votes, servers may stay in the candidate state for multiple elections.}
\label{fig:raft_states}
\end{SCfigure}

Designed for understandability, Raft relies on a leader to attain consensus under all non-Byzantine conditions. Consensus in Raft consists of two major phases: leader election and log replication. The system first elects a leader server by leader election, and then the leader conducts consensus in log replication. We now introduce some key concepts and limitations in Raft leader election. 

\subsection{Raft Basics}
Raft deploys three servers states, and each server assumes one of the three states at any given time: leader, follower, and candidate. Under normal operation, there is only one leader and other servers assume a follower role. Followers only passively respond to requests, whereas the leader dominates the consensus by forcing followers to synchronize its log. Thus, the leader becomes the only bridge between internal servers and external clients. The leader receives entries from clients and issues \texttt{AppendEntries RPCs} to all the other servers. If the leader collects replies from a majority of servers, then the consensus for committing the entries is reached. 

\begin{figure*}[t]
\minipage{0.31\textwidth}
    \centering
    \begin{adjustbox}{width=\linewidth}
        \includegraphics 
        {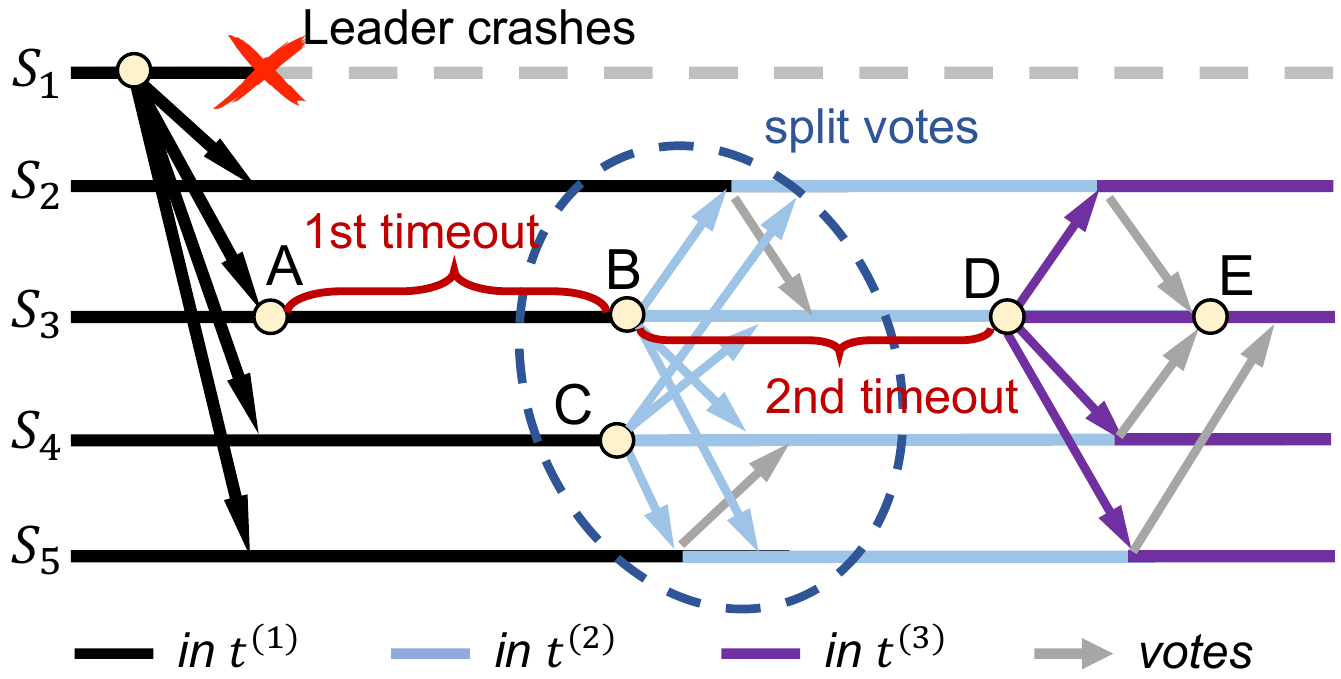}
    \end{adjustbox}
    \caption{Split votes with concurrent election campaigns. A new leader is elected after one of the competing candidates triggers a second timeout.
    }
    \label{fig2:competing_candidates}
\endminipage \hfill
\minipage{0.33\textwidth}
    \centering
    \begin{adjustbox}{width=\linewidth}
        \includegraphics{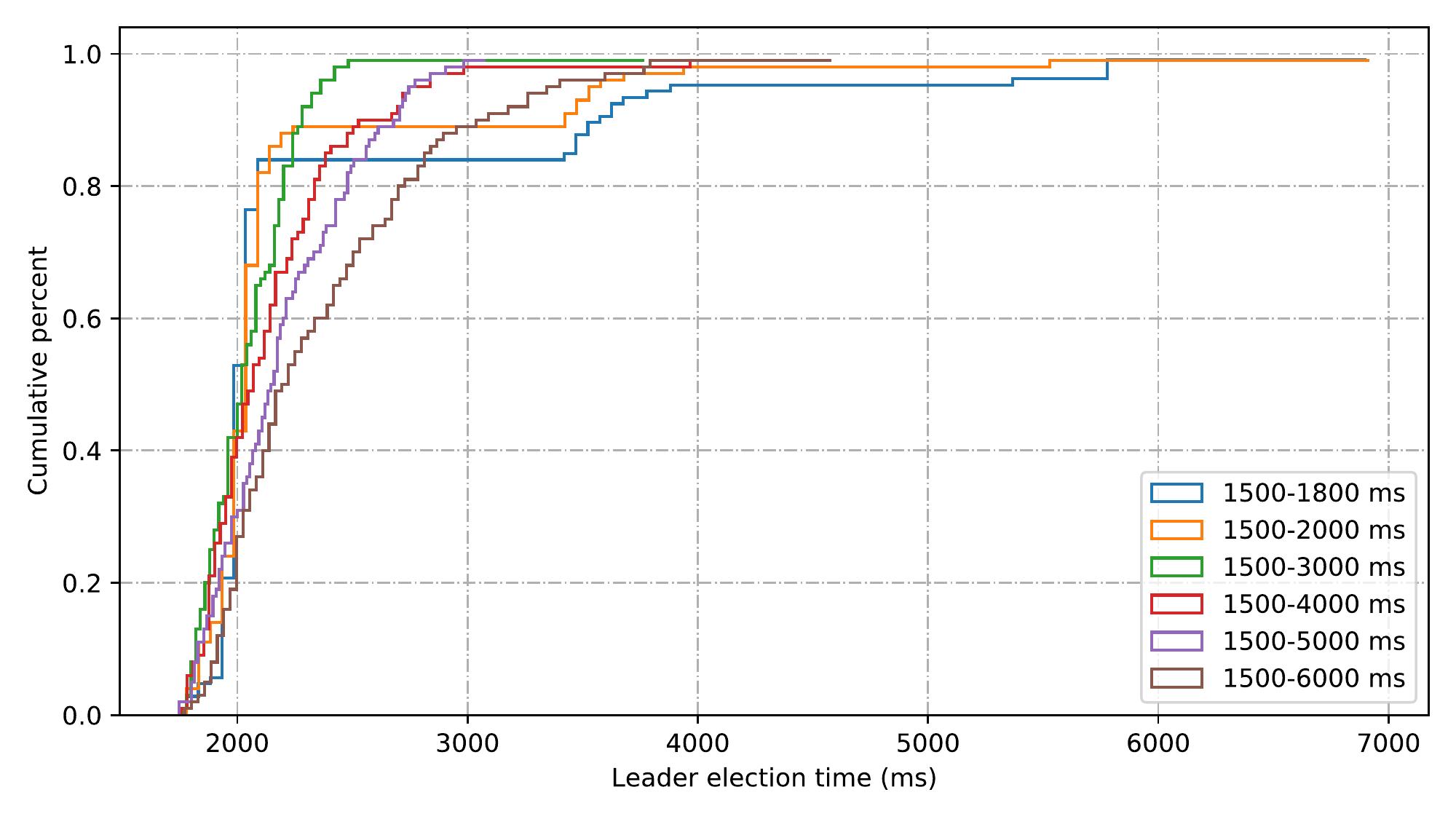}
    \end{adjustbox}
    \caption{Leader election time in a $5$-server Raft cluster under varying amounts of randomness of election timeouts.}
    \label{fig:raft_le_election_timeouts}
\endminipage \hfill
\minipage{0.33\textwidth}
    \centering
    \begin{adjustbox}{width=\linewidth}
        \includegraphics{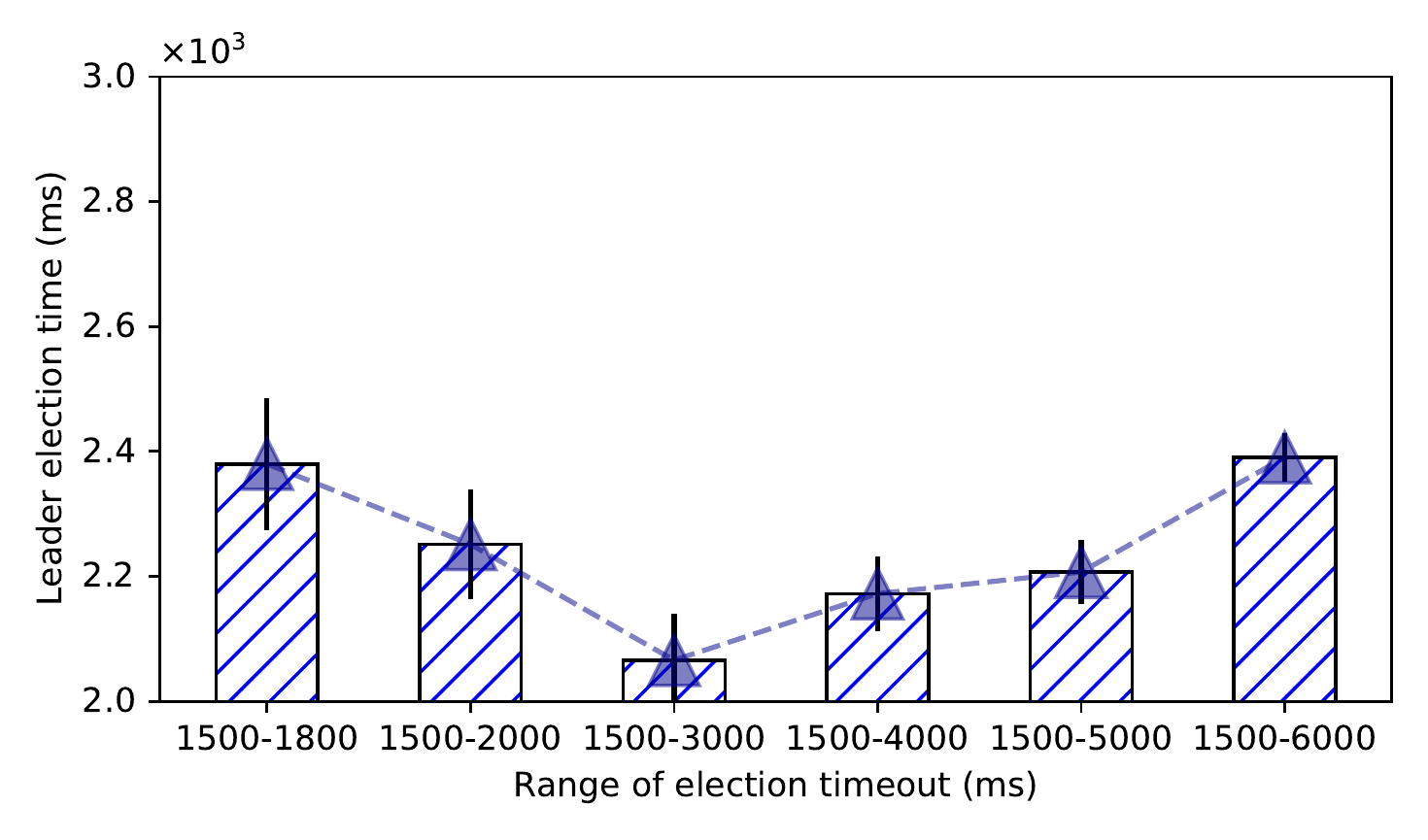}
    \end{adjustbox}
    \caption{Averaged leader election time with increasing amounts of randomness of election timeouts.}
    \label{fig:raft_le_et_bar}
\endminipage
\end{figure*}

In Raft, the logical time is represented by \emph{terms}. Terms are positive integers that increase monotonically. Each term begins with a leader election period and proceeds with a subsequent log replication period if a leader is elected. A server never reduces its term, which prevents the server from receding to a previous time point that could overwrite committed entries. In leader election, a new leader must have the highest term and log entries consistent with a majority of servers. This mechanism forces the system to always elect the most up-to-date server as a leader.

If an incumbent leader crashes, the consensus service is suspended until a new leader is elected. Raft uses timeouts to detect leader failures. Initially, each server joins the system as a follower and starts a timer. A server resets its timer when receiving a heartbeat from the leader. If the leader crashes, a follower will trigger a timeout and transition to a candidate; it increments its term and initiates a leader election campaign by broadcasting \texttt{RequestVotes RPCs} to solicit votes. If it collects votes from a majority of servers (including itself), the candidate becomes the next leader (shown in Figure~\ref{fig:raft_states}). However, if the candidate cannot collect enough votes within the election timeout nor receive a higher-termed message to transition back to a follower, the candidate initiates a new election campaign and repeats the above mentioned process. Other servers vote for a candidate if it satisfies the three following requirements:
1) the candidate's term is not less than the follower's term;
2) the follower has not voted yet in its term;
3) the candidate's log is at least as up-to-date as the follower's log.


\subsection{Split Votes in Leader Election}
\label{sec:split_votes}

The election regime may split votes in a given term. By following the three requirements, votes can be inopportunely split among candidates when concurrent election campaigns take place. Consequently, no candidate can collect votes from a majority of servers to be elected.

Figure~\ref{fig2:competing_candidates} shows an example of split votes in a Raft cluster of $5$ servers, in which terms are denoted by $t$. Assume server $S_1$ was the leader in $t^{(1)}$ (i.e., $term{=}1$) and crashed. Afterwards there was no communication between $S_1$ and the other servers. Consequently, timers of $S_3$ and $S_4$ expire and trigger new election campaigns at points $B$ and $C$, respectively. Meanwhile, $S_3$ and $S_4$ increase their terms to $t^{(2)}$ and broadcast leader election requests. Then, $S_2$ receives an election request with $t^{(2)}$ from $S_3$ before it receives the one from $S_4$. Since each follower only votes once per term, $S_2$ votes for $S_3$ and ignores the request from $S_4$ in $t^{(2)}$. On the contrary, $S_5$ votes for $S_4$ and denies the request from $S_3$. At this time, since neither $S_2$ nor $S_4$ can collect three votes to win, the two candidates are trapped, and a split-votes scenario occurs.

The system has to wait for new elections initiated from servers whose timers expire. At point $D$, $S_3$'s timer expires again, and thus $S_3$ initiates a new election campaign while increasing its term to $t^{(3)}$. Fortunately, $S_3$ succeeds in collecting votes from a majority of servers without disruptions from the other servers. Then, $S_3$ is elected to be the new leader at point $E$ and starts conducting log replication in $t^{(3)}$. After receiving heartbeats issued from the new leader, $S_4$ steps back to follower. From then on, the system recovers from the OTS condition and resumes to normal operation.

In addition, when different groups of servers have a lower internal (in-group) network latency but a higher external (between groups) network latency (e.g., geographically distributed servers), systems are more susceptible to split-vote scenarios. In this case, servers communicate faster within their ``local'' group than with the external servers. Thus, a candidate is more likely to succeed in collecting votes from its own group, and election requests from outside-group candidates will be repeatedly ignored if followers have voted for their ``local'' candidate. This problem also arises in distributed transient networks. Huang et al.~\cite{huang2019performance} hinted that network split and message loss often cause multiple elections, which exacerbate the delay of the emergence of a new leader.

\section{Problem Statement}
\label{sec:timer-randomization-tradeoff}
Raft uses timer randomization to alleviate split votes. We implemented Raft and evaluated its leader election performance in a $5$-server cluster. We measured $1000$ runs for each of the six ranges of election timeouts. The network latency was set to $100$-$200$~ms. Figure~\ref{fig:raft_le_election_timeouts} presents the results in cumulative distribution of percentage, and Figure~\ref{fig:raft_le_et_bar} shows the average election time. The results indicate that adding randomization can help the system to avoid split votes. A successful leader election campaign comprises two periods: the detection of the absence of leadership and election of a new leader. The former one is achieved by timer expiration with randomized election timeouts, and the latter one is accomplished by collecting votes from a majority of servers. 

However, there exists a tradeoff between the durations of detection and election, which are associated with timer randomization. In particular, if the range of the randomization is extended, the positive effect is that it can shorten the election period but extends the detection period. A higher amount of randomness can alleviate split votes because the election timeouts of timers have a higher chance to differ from each other. This reduces the chance that multiple candidates simultaneously initiate election campaigns to solicit votes. While followers only passively respond to other servers, the negative effect is that it takes longer for servers to detect the absence of leadership, thereby increasing the leader election time. Conversely, if the amount of randomness shrinks, though it takes a shorter amount of time to detect a crashed leader, the probability of competing candidates increases. Servers are more likely to simultaneously initiate leader election campaigns, which inflict additional phases in leader election and prolong the emergence of a new leader.

For example, in Figure~\ref{fig:raft_le_election_timeouts}, with $300$~ms of randomness from the range of $1500{-}1800$~ms, due to split votes, approximately $18\%$ of the election campaigns cannot converge in $3500$~ms. The split-vote situation is mitigated with a larger amount of randomness. For example, with $500$~ms of randomness from the range of $1500{-}2000$~ms, less than $12\%$ of the election campaigns experience split votes, and according to Figure~\ref{fig:raft_le_et_bar}, the average election time decreases. However, with increasing amount of randomness, the election time rises, in which the detection of the absence of leadership takes longer than the election of a new leader.


In clusters at large scales, we can set even larger amounts of randomness to avoid split votes, but this may not benefit to shorten the OTS time. If we could avoid split votes without sacrificing the duration of detection or election, the performance of leader election can be improved from both sides, thereby enhancing the reliability of services.


\section{\scap Leader Election}
\label{sec:dple}
In this section, we introduce \scap, a leader election protocol that takes precautions against leader failures by preparing a pool of ``potential leaders''. These potential leaders are furnished with configurations that enable them to win in future election campaigns. A configuration contains a unique priority and an election timeout, and configurations are dynamically assigned to servers.

In general, \scap enables fast leader election by always assigning a configuration that leads to an undefeated campaign to a server that has the most potential to become the next leader. To achieve this goal, \scap assigns unique configurations to different servers via stochastic configuration assignment (SCA). In normal operation, \scap dynamically rearranges and distributes these configurations through the probing patrol function (PPF).

\subsection{Stochastic Configurations Assignment}
\label{sec:stochastic_config_assinments}
The main idea of SCA is to enable differential growth rates of terms when multiple election campaigns take place. 
Let us first consider all candidates are qualified to assume future leadership; that is, they have the same up-to-date log states and are in the latest term (disqualified candidates will be considered in Section~\ref{sec:probing_patrol}). 
By applying differential growth rates of terms, if followers trigger timeouts simultaneously, their election campaigns will take place in different terms and will not flock into a single term to form \textit{flocked elections} that compete against each other as in Raft. The candidate with the highest-term campaign wins and becomes the next leader.

\subsubsection{Stochastic configurations}
To achieve stochastic configuration assignment, when joining the system, each server $S_i$ assigns a unique priority, denoted by $\mathcal{P}_i$. The priority determines the term growth and election timeouts. A configuration is denoted by
$\pi^{(\mathcal{P}_i)}$ where $\mathcal{P}_i$ is the priority. 
To maintain simplicity, \scap implements the priority by using server IDs. A server assigns a unique server ID, $i$, as its initial priority, where $i$ is not assigned to any other server; that is, $\forall S_i, S_j \in \scap, i \neq j$. 
A server's ID determines its priority such that $\mathcal{P}_i = i$.

\subsubsection{Election timeouts}
A server's priority initializes its election timeout as Eq.~\ref{eq:1}.
\begin{equation}
\label{eq:1}
period_i = baseTime + k(n - \mathcal{P}_i)
\end{equation}
In Eq.~\ref{eq:1}, $baseTime$ is of constant value that should be set significantly larger than the network latency (e.g., if the network latency is $10$ms, we can set $baseTime {=} 100$ms). $k$ is a constant value (ms) for adjusting the election timeout interval; a higher $k$ promotes a larger gap among server election timeouts. $n$ is the number of servers in the system.
For example, in a $10$-server cluster with a setting of $baseTime {=} 100$~ms and $k {=} 10$, server $S_2$'s initial priority $\mathcal{P}_2 {=} 2$, so its election timeout $period_2 {=} 180$ms. For server $S_{10}$ ($\mathcal{P}_{10} {=} 10$), its initial election timeout is the base time ($100$~ms). 

A configuration with a higher priority is paired with a shorter election timeout. Normally, with a shorter election timeout, a server often detects leader failures faster since its timer soon runs out of time. This regime provides an optimization that a higher-priority server will have a quicker response to leader failures.

\begin{figure*}[t]
\captionsetup[subfigure]{}
\centering
\subfloat[
All servers synchronously respond to the leader in the $k$-th configuration heartbeat, but $S_4$ and $S_5$ fall behind in log replication. Then, the higher-priority configurations previously possessed by $S_4$ and $S_5$ are rearranged to more up-to-date servers ($S_2$ and $S_3$).]
  {\includegraphics[width=0.45\textwidth]{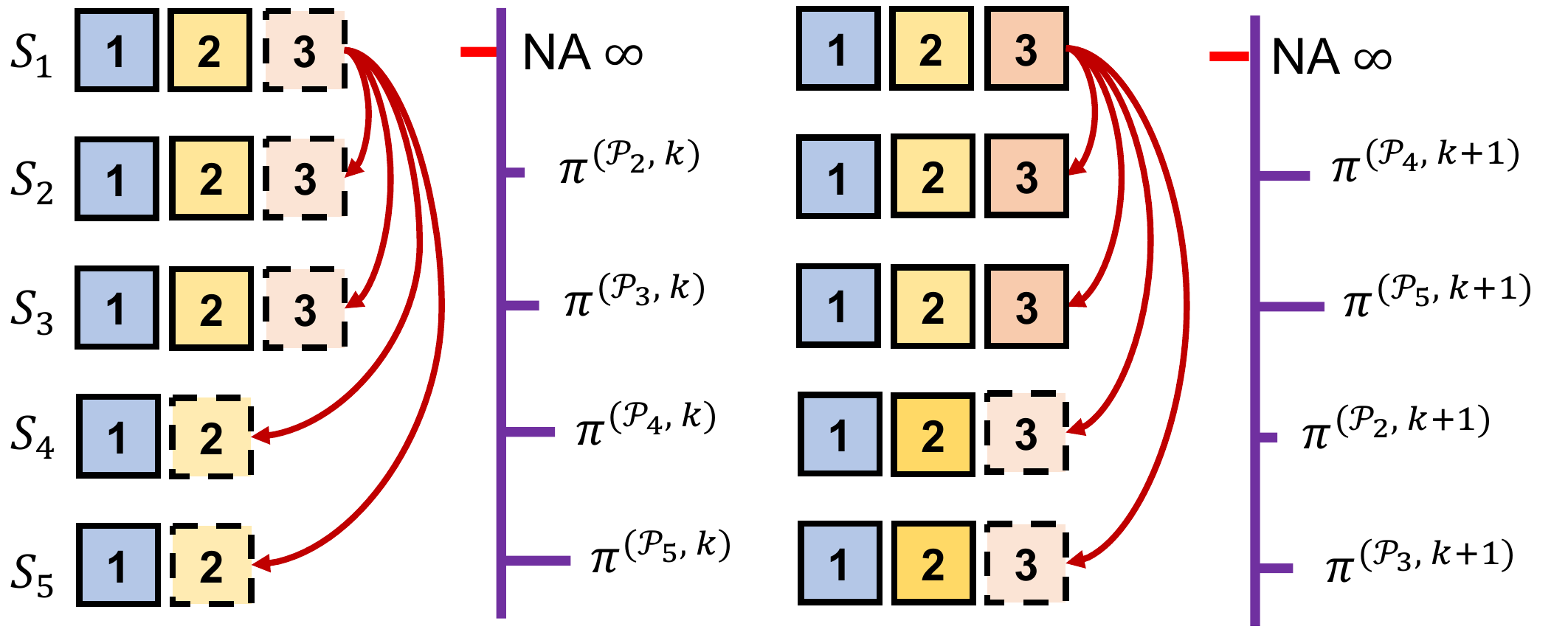}%
\centering
\label{fig:patrol-a}}
\hfil 
\subfloat[
$S_2$ and $S_4$ crash after receiving the $k$-th configuration heartbeat. Then, $S_2$ recovers in the $k{+}1$-th heartbeat, but $S_3$ is still crashed. After the rearrangement, since $S_4$ cannot receive the newly issued configuration, $S_4$ will have a stale configuration after recovery.]
  {\includegraphics[width=0.45\textwidth]{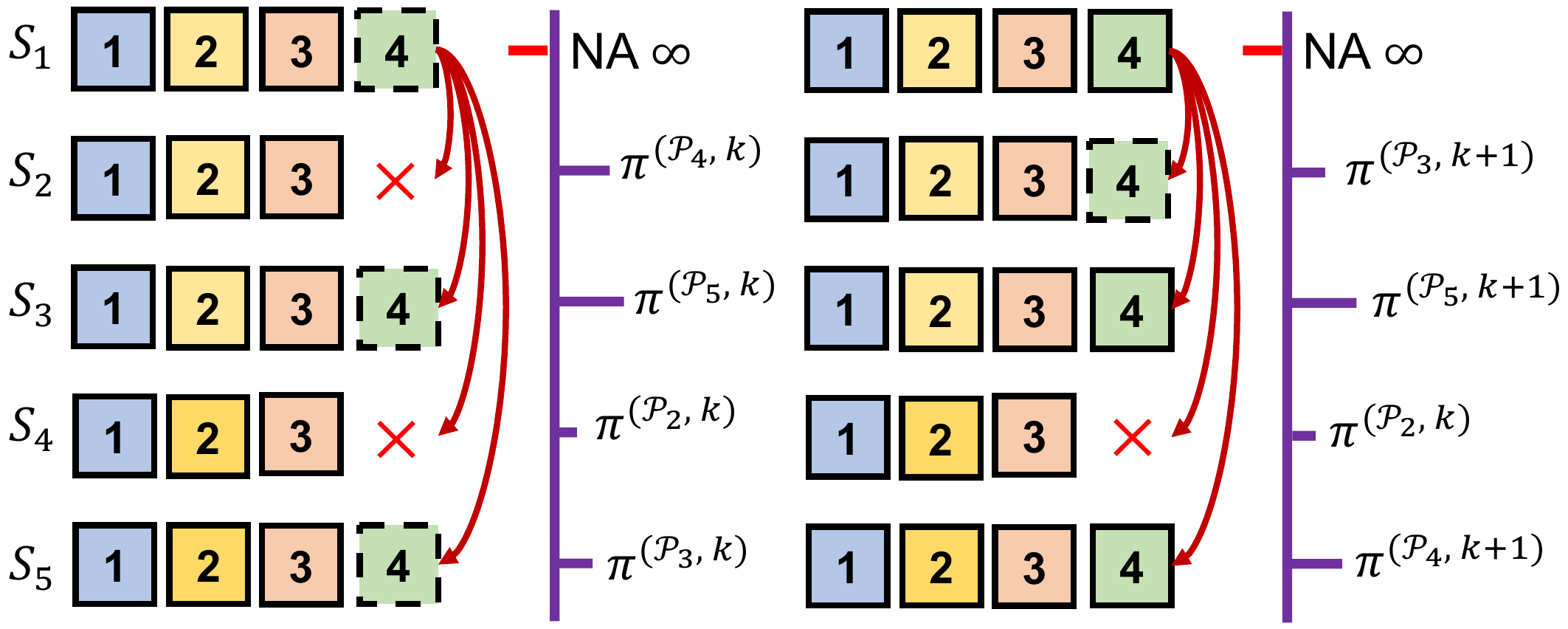}
\label{fig:patrol-b}}
\caption{Examples for the probing patrol function in a 5-server \scap system.}
\label{fig:patrol}
\end{figure*}

\subsubsection{Term growth}
A server's priority also determines its term growth. In \scap, we denote a server $S_i$'s term as $\mathcal{T}_{S_i}$. If $S_i$ initiates a new leader election campaign, $S_i$ increases $\mathcal{T}_{S_i}$ by Eq.~\ref{eq:2}:
\begin{equation}
\label{eq:2}
\mathcal{T}_{S_i}^{(k+1)} \leftarrow \mathcal{T}_{S_i}^{(k)} +  \mathcal{P}_{i},
\end{equation}
In addition, a server always updates its term when it receives messages with a higher term from other servers. The update is in accordance with Eq.~\ref{eq:3}.
\begin{equation}
\label{eq:3}
    \mathcal{T}_{S_i}^{(k+1)} \leftarrow max(\mathcal{T}_{S_i}^{(k)}, \mathcal{T}_{S_j, j \neq i}^{(k+1)}),
\end{equation}
where $\mathcal{T}_{S_i}^{(k)}$ is the term of the latest update, and $\mathcal{T}_{S_j, j \neq i}^{(k+1)}$ is the received term from other servers. $S_i$ always sets its term to the highest value regardless of the other arguments in a received message. For example, let us consider a server $S_i$ with a priority $\mathcal{P}_i=2$ working in term $3$. If $S_i$'s timer expires, $S_i$ initiates an election campaign and sets its term to $5$ (Eq.~\ref{eq:2}). If $S_i$ receives a message sent from a server $S_j$ whose term is $4$, $S_i$ simply ignores the message. However, if $S_j$'s term is $6$, $S_i$ sets its term to $6$, regardless of other information in the message (Eq.~\ref{eq:3}). 

In contrast to Raft, the increment of logical time in \scap is no longer consecutive: simultaneously initiated election campaigns are scattered into different terms. This avoids flocked elections and split votes before a leader appears. Similar approaches have been adapted by some state-of-the-art systems, such as Redis LE~\cite{redisLeaderElection} and Zookeeper~\cite{zookeeperleaderelection}.
However, SCA alone does not always lead the system to the desired condition since the network may change unpredictably. Current up-to-date servers may fall behind in the next round, whereas slower servers may catch up. Therefore, a mechanism that rearranges configurations based on server log responsiveness is necessary to keep the election scheme efficient.

\subsection{Atomic Configuration Rearrangement}
\label{sec:probing_patrol}

A new leader is the candidate that has successfully collected votes from a majority of servers in the highest term. Yet, server configurations are not related to log responsiveness (i.e., log replication status). A tricky scenario may occur: if the candidate that has the highest term does not have an up-to-date log, its election campaign cannot succeed because up-to-date servers never vote for it; other servers resume their election campaigns after updating their terms to the highest one. Therefore, the configuration that empowers a high term growth is ``wasted'' if the candidate is not an up-to-date server, which undermines the purpose of applying SCA.


To tackle this challenge, \scap reassigns configurations that are inclined to win a future election to servers that have up-to-date logs through the probing patrol function (PPF). First, PPF keeps track of servers' log index, which indicates a server's log responsiveness, via the periodical heartbeat sent from the leader (Listing~1). Then, the leader collects replies and rearranges configurations to each server. Next, it broadcasts new configurations in the following heartbeat. Finally, servers update to the newly assigned configuration if the received one is different. 

PPF assists the redistribution of configurations based on log responsiveness, yet the process is not atomic because configurations are not atomically rearranged. For instance, if server $S_i$ with configuration $\pi^{(\mathcal{P}_i)}$ crashes, the leader may assign $\pi^{(\mathcal{P}_i)}$ to another server (say $S_j$). When $S_i$ recovers, $S_i$ restores its stale priority ($\pi^{(\mathcal{P}_i)}$); at this time, if the leader crashes, $S_i$ and $S_j$ may trigger election campaigns simultaneously and split votes.

\definecolor{mygray}{rgb}{0.4,0.4,0.4}
\definecolor{mygreen}{rgb}{0,0.8,0.6}
\definecolor{myorange}{rgb}{1.0,0.4,0}

\lstset{
basicstyle=\footnotesize\ttfamily\color{black}\bfseries,
commentstyle=\color{mygray},
frame=single,
numbersep=10pt,
numberstyle=\tiny\color{mygray},
keywordstyle=\color{blue},
showspaces=false,
showstringspaces=false,
stringstyle=\color{myorange},
tabsize=2
}

\begin{lstlisting}[language=GO, caption={Parameters in \scap compared with Raft},frame=none,escapechar=!]
//the parameters of AppendEntries RPCs
type AppendEntriesArgs struct{
    term         int64
    leaderId     string
    prevLogIndex int64
    prevLogTerm  int64
    entries[]    Entries
    leaderCommit int64
    !\tikzmark{a}!newConfig    Configurations //newly added!\tikzmark{b}!
}

type Configurations struct{
    timerPeriod time.Duration
    priority    int64
    confClock   int64
}

//the reply messages
type AEReplyArgs struct{
    term      int64
    success   bool
    !\tikzmark{c}!status    configStatus //newly added!\tikzmark{d}!
}

type configStatus struct{
    LogIndex    int64
    timerPeriod time.Duration 
}
\end{lstlisting}
\begin{tikzpicture}[remember picture, overlay]
\draw[red] ([shift={(-3pt, 1.5ex)}]pic cs:a) rectangle   ([shift={(3pt,-0.65ex)}]pic cs:b);
\draw[red] ([shift={(-3pt, 1.5ex)}]pic cs:c) rectangle   ([shift={(3pt,-0.65ex)}]pic cs:d);
\end{tikzpicture}

\vspace{-1.8em}

To prevent stale configurations from interfering with leader election, \scap adds a hyperparameter, configuration clock (denoted by \texttt{confClock} in Listing~1), that shows the freshness of configurations. The configuration clock is the logical clock of configuration rearrangements. In a configuration $\pi^{(\mathcal{P}_i, k)}$, $k$ denotes the configuration clock of this configuration, $\pi$. In normal operation, the configuration clock increments monotonically with the number of heartbeats. In leader election, \emph{servers never vote for candidates whose configuration clock is stale}. That is, a candidate's configuration clock should not be less than a voter's configuration clock. As for deciding a new leader, \scap adopts all the rules of Raft.

Figure~\ref{fig:patrol} shows two examples of configuration rearrangements in a $5$-server cluster. In Figure~\ref{fig:patrol-a}, initially, server configurations are assigned in accordance with Eq.~\ref{eq:1}. Since PPF rearranges configurations according to log responsiveness,
if higher-priority servers fall behind in log replication, servers with the same log as the leader (most up-to-date logs) are assigned with higher-priority configurations.

When a server recovers (e.g., from a crash or network issues), its configuration may become stale because the server is unable to sync up with the leader. For example, in Figure~\ref{fig:patrol-b}, $S_2$ and $S_4$ crash and do not respond to the leader after they receive the $k$-th configuration heartbeat. Then, $S_2$ recovers in the next heartbeat, but $S_4$ does not. $S_2$'s high-priority configuration is assigned to $S_5$, though $S_5$ still has a fresh configuration. Unfortunately, $S_4$'s configuration becomes stale. After $S_4$ recovers, $S_4$'s \texttt{confClock} is $k$, but the system already operates in a higher configuration clock. $S_4$ needs to synchronize with the leader to refresh its configuration.

The probing patrol function enables a resilient configuration distribution and promotes the efficiency of leader election since it avoids that a higher-priority configuration is assigned to a losing candidate (i.e., a server without the most up-to-date log). Therefore, when the current leader crashes, the server with the highest-priority configuration has the maximum potential to detect the leader failure and initiate a new election campaign before any other servers.


\subsection{Declining Competing Candidates}

The combination of SCA and PPF enables fast leader election that avoids split votes. When candidates start new elections simultaneously, the election with the highest term always defeats other elections. After followers synchronize to the highest term, they will refuse to respond to lower-term requests; thus, votes in a given term aggregate at one server, terminating the election in a single election campaign.

\begin{figure}[t]
    \centering
    \includegraphics[width=0.4\textwidth]{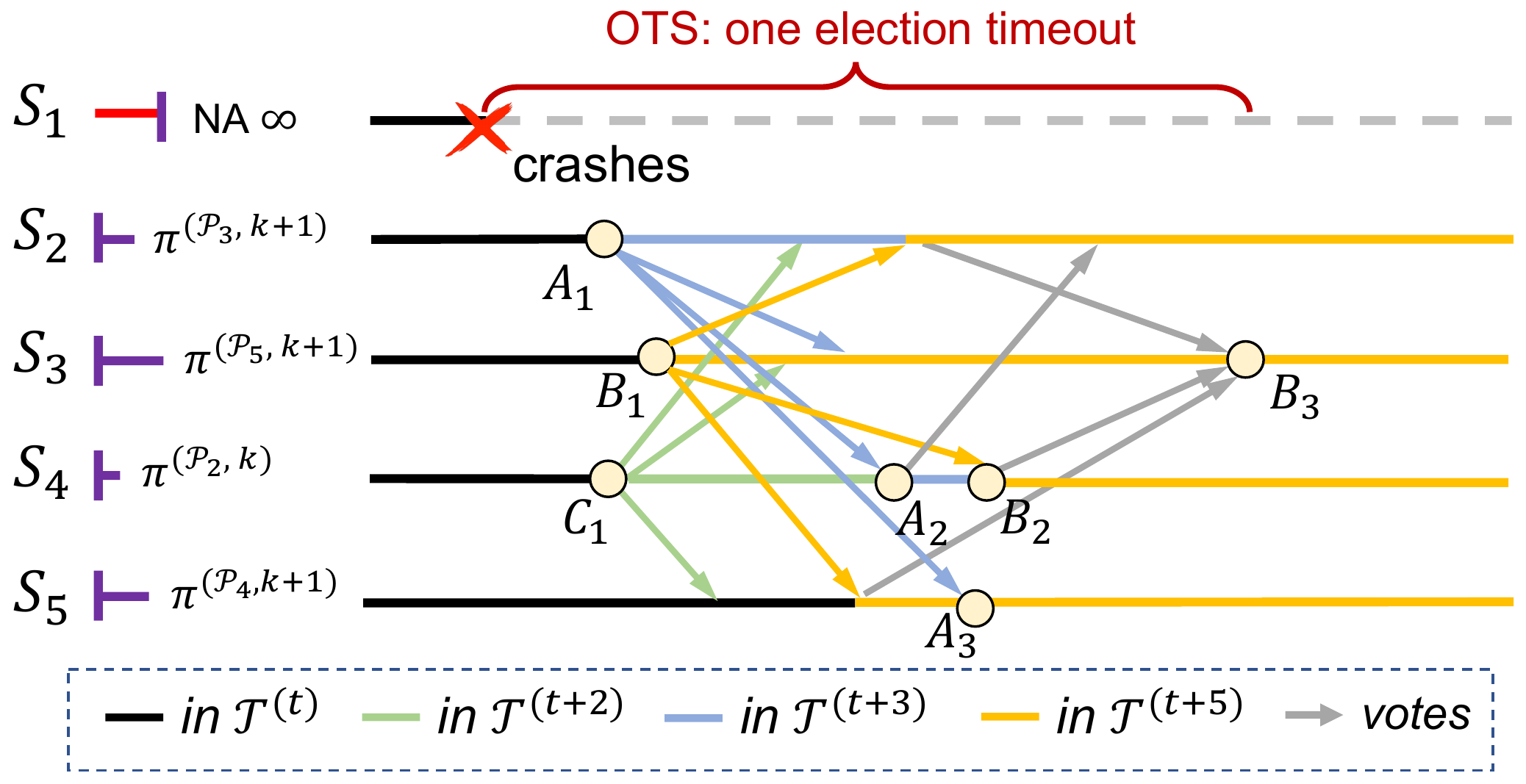}
    \caption{Concurrent election campaigns, initiated simultaneously,
      always have a discrepancy in terms; \scap\ leads to faster
      convergence by electing a higher termed leader.}
    \label{fig:scap_leader_election}
\end{figure}

Suppose the leader in Figure~\ref{fig:patrol-b} crashes in term $t$, and we show an extreme case (Figure~\ref{fig:scap_leader_election}) to illustrate how \scap elects a leader without competition when three concurrent leader election campaigns materialize. 
After $S_2$, $S_3$, and $S_4$'s timers expire, three election campaigns are initiated simultaneously at point $A_1$, $B_1$, and $C_1$, respectively. In this case, no server votes for $S_4$ since $S_4$'s configuration is stale. 
On the contrary, $S_2$ and $S_3$ are up-to-date servers, but $S_3$'s election campaign has a higher term ($t{+}5$) than $S_2$ (in term $t{+}3$). Thus, $S_4$ synchronizes to term $t{+}3$ at point $A_2$ and to term $t{+}5$ at point $B_2$, while $S_5$ synchronizes directly to term $t{+}5$ and then rejects $S_2$'s request at point $A_3$. $S_3$'s election campaign takes over the other campaigns, and $S_3$ is elected at point $B_3$, when votes from a majority of servers are collected. Overall, leader election converges in one election campaign, and the system recovers after $S_3$ becomes the leader.

\scap maintains the understandability of Raft and can also be easily deployed. If the system is fully synchronous, the configuration rearrangement can be implemented by a separate heartbeat at a low interval rate. The separation can reduce the messaging cost. The leader in \scap needs to sort and assign configurations, a task with linear time complexity, imposing a slight computational cost.

The idea of \scap differs from the optimization that Raft intended to apply to solve split votes. The authors of Raft argued that they tried to rank candidates in the same term to solve split-vote scenarios and implied that this measure would cause corner cases that undermine safety~\cite{ongaro2014search}. Figure~\ref{fig:op_compare} shows an intuitive comparison of the ideas behind Raft and \scap where three candidates initiate election campaigns. Term surfaces represent the logical times a system currently accrues. If we draw a term surface for candidates (not considering followers), Raft leader election campaigns are more likely to lie in the same surface (left in Figure~\ref{fig:op_compare}). Raft's authors intended to rank candidates in the same surface to intensify the competition among candidates in a given term, thereby the additional mechanisms may cause corner cases that put log safety in peril. However, due to the discrepant term growth enabled by prioritized configurations in \scap, simultaneously initiated election campaigns escape from competing on the same surface but reside on their own term surfaces, and the system leaps to the highest surface with a new leader. Therefore, concurrent campaigns are in a cube ($3D$ vs. $2D$) consisting of varying term surfaces, and the candidate on the top surface defeats the other campaigns, leading the system to the highest term.

In addition to the use in Raft, \scap can also be adapted to other consensus protocols. 
Since the correctness of leader-based consensus protocols relies on the combination of leader election (view-change) and log replication phases, there is no one-size-fits-all leader election mechanism that stands alone to be directly applicable in every consensus protocol. 
However, the concept of \scap can be adopted to optimize leader-based consensus protocols in terms of preparing ``future leaders'' in advance, before the next leader election takes place. 
For example, the slave election and promotion in Redis~\cite{redisLeaderElection} ranks slave servers with \texttt{SLAVE\_RANK}, which determines the delay for a slave server to get elected in a given \texttt{configEpoch}. This approach has similar issues as Raft's competing candidates discussed in Section~\ref{sec:background}. Slave servers may simultaneously initiate elections and reside in the same \texttt{configEpoch}. Applying \scap can avoid potential competitions and boost the election process. The time spent on leader election is something that leader-based systems want to avoid but have to endure; similar cases such as~\cite{azurele} and~\cite{hunt2010zookeeper} can also apply \scap to enable fast leader election, paving the way to avoid potential future conflicts.

\begin{figure}[t]
    \centering
    \includegraphics[width=0.35\textwidth]{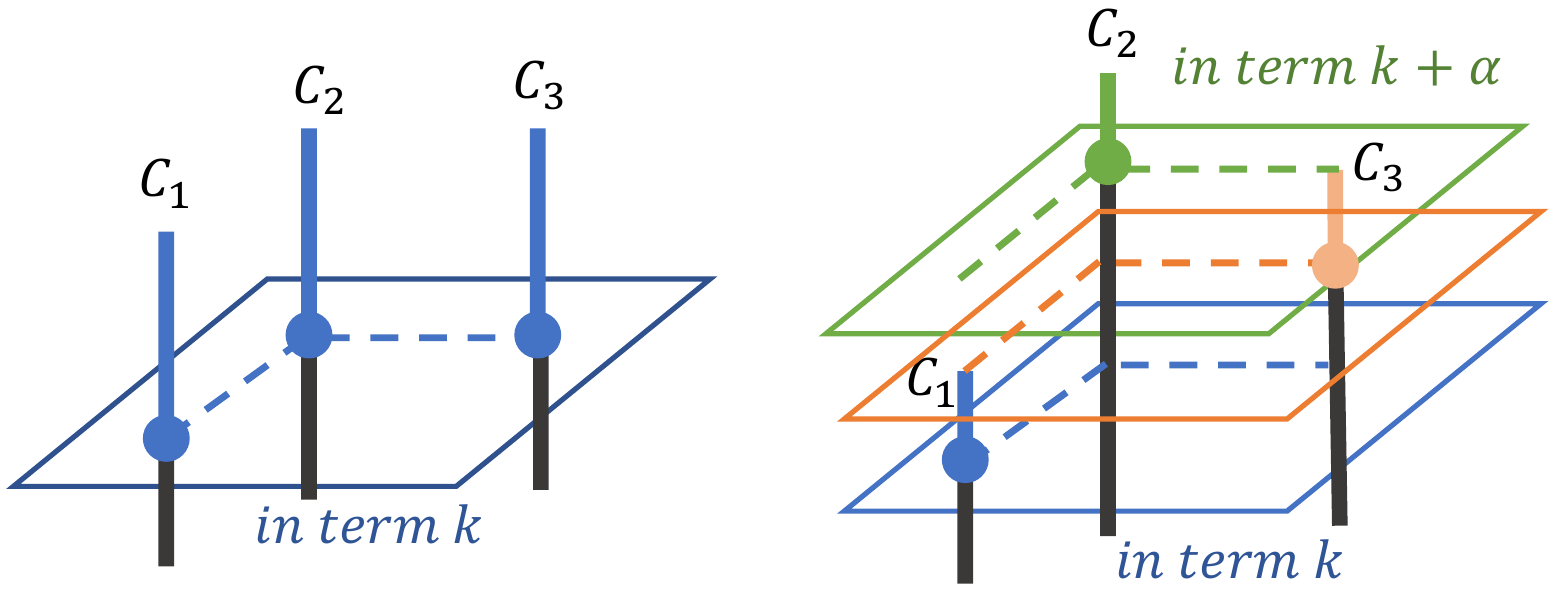}
    \caption{Raft intended to rank competing candidates whose campaigns are in the same term (left), while \scap involves priority-based configuration assignments to distribute concurrent campaigns into different terms (right).}
    \label{fig:op_compare}
\end{figure}

\section{Correctness Arguments}
\label{sec:correctness}
While accelerating the completion of leader election by avoiding split votes, \scap maintains Raft's safety and liveness properties. \scap also exhibits a stronger liveness than Raft. We discuss these points in detail in the remainder of this section.

\begin{lemma}
\label{lemma:translation}
An \scap leader election is able to be translated to Raft leader elections.
\end{lemma}

\begin{proof}
An \scap leader election can be implemented by multiple consecutive Raft leader elections depending on a server's priority. We denote an \scap leader election initiated by a server with a priority $\mathcal{P}$ as $E_{\mathcal{P}}$. We assume the system is in term $t$ (i.e., $\mathcal{T}^{(t)}$) before the election takes place. Consequently, the system will be in term $t+\mathcal{P}$ after the server wins the election and becomes the new leader. We show that $E_{\mathcal{P}}$ can be implemented by $\mathcal{P}$ consecutive Raft leader elections whose term increases from $t$ to $t+\mathcal{P}$, denoted by $\mathcal{R}_{t \rightarrow t+\mathcal{P}}$; that is,
$$ \mathcal{R}_{t \rightarrow t+\mathcal{P} } \Longrightarrow E_{\mathcal{P}}, \text{\ where \ } |\mathcal{R}_{t \rightarrow t+ \mathcal{P}}| = \mathcal{P}$$

Suppose a server $S$ operating in term $t$ in Raft triggers a timeout and becomes a candidate. $S$ then increases its term to $t+1$. At this time, if $S$ loses its connection to other servers, $S$ will time out again and increase its term to $t+2$. We denote this period of disconnection as a \emph{blackout window}. Raft allows a sufficiently large blackout window because it tolerates all non-Byzantine (benign) failures. Thus, a server's leader election with a priority $\mathcal{P}$ in \scap is implemented by $\mathcal{P}$ leader elections in Raft that take place in a blackout window. \end{proof}

Figure~\ref{fig:liveness_proof} illustrates an example of Lemma~\ref{lemma:translation}. $S_i$ and $S_j$ are two \scap servers. Before their elections, $S_i$ and $S_j$ are in term $\mathcal{T}^{(t)}$ with configurations of $\pi^{(P_3)}$ and $\pi^{(P_2)}$, respectively. After initiating their election campaign, $S_i$'s term increases directly to $\mathcal{T}^{(t+3)}$ (at $A_3$) while $S_j$'s term increases to $\mathcal{T}^{(t+2)}$ (at $B_2$). These two election campaigns can be implemented by Raft leader elections, where $S_i$ has a blackout window of three consecutive timeouts (at $A_1$, $A_2$, and $A_3$), and $S_j$ has that of two consecutive timeouts (at $B_1$ and $B_2$).

\begin{theorem}[Validity]
If all the servers have the same input, then any value decided upon must be that common input.
\end{theorem}

\begin{proof}
The validity property holds in Raft~\cite{ongaro2014search}, where a committed value must have been logged by a majority of servers. With Lemma~\ref{lemma:translation}, since \scap leader election can be implemented by Raft leader elections, the improvement made by \scap is isolated from log replication in Raft. Thus, \scap maintains the validity property.
\end{proof}

\begin{figure}[t]
    \centering
    \includegraphics[width=0.9\linewidth]{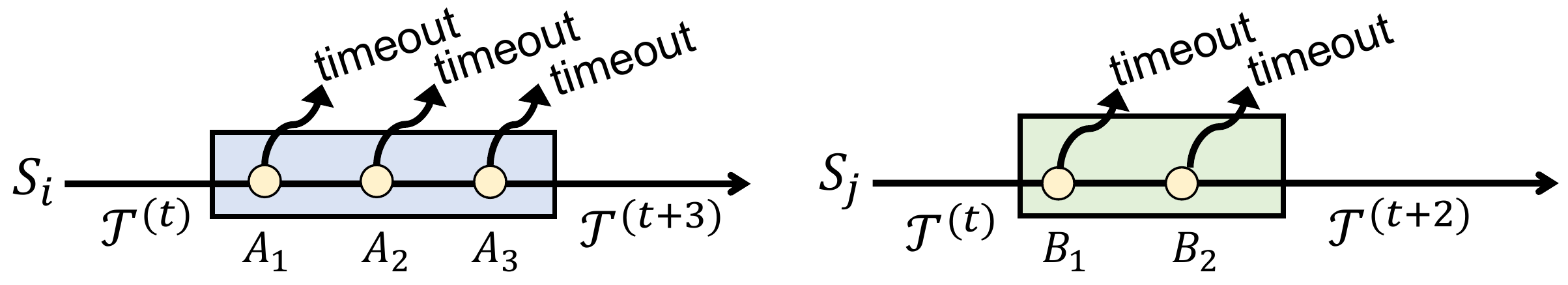}
    \caption{An \scap leader election can be translated into multiple consecutive Raft leader elections.}
    \label{fig:liveness_proof}
\end{figure}

\begin{lemma}
\label{lemma:indistinguishable}
An \scap election campaign and a Raft election campaign initialized by the same server is indistinguishable on other servers.
\end{lemma}

\begin{proof}
Suppose a cluster runs Raft and \scap simultaneously where a server, say $S_i$, initiates an \scap leader election campaign and sends a request to another server, say $S_j$. We show that $S_j$ cannot distinguish if this request is from a Raft election campaign, denoted by $R_{S_i}$ or an \scap election campaign, denoted by $E_{S_i}$. That is,
$$\forall S_j \in \text{Raft}  \wedge \scap ,  E_{S_i} \mathop\sim^{S_j} R_{S_i}.$$

\scap does not change the rules or the parameters of election requests for electing a new leader in Raft. From Lemma~\ref{lemma:translation}, an \scap election can be implemented by multiple Raft elections, and $S_j$ is not aware of the length of the blackout window on $S_i$. Therefore, when $S_j$ receives the \scap election campaign request from $S_i$, $S_j$ is unable to distinguish if the request from an \scap election or a Raft election at the end of the blackout window, i.e., for $S_j$, $E_{S_i} \sim R_{S_i}$.
\end{proof}

\begin{theorem}[Safety]
Nonfaulty servers do not decide on conflicting values.
\end{theorem}

\begin{proof}[Proof (Sketch)]
The safety property also holds in Raft~\cite{ongaro2014search}. \scap leader election is isolated from the log replication phase. Thus, under normal operation, \scap maintains the safety property. In addition, during a view-change/leader election phase, from Lemma~\ref{lemma:indistinguishable}, \scap and Raft election campaigns are indistinguishable, and servers follow the same rule to determine a new leader. Thus, a new leader elected by \scap must be able to be elected by Raft with a specific initial configuration. A leader in \scap maintains its correctness and no committed values can be overwritten. Therefore, \scap maintains Raft's safety property.
\end{proof}

\begin{lemma}
\label{lemma:1}
If any two servers are in the same configuration clock, they must possess different configurations.
\end{lemma}

\begin{proof}
Initially, SCA assigns each server a unique configuration based on the server ID. Under normal operation, PPF atomically rearranges configurations, based on server log responsiveness, at any given configuration clock. If the leader fails, the assignment task may become incomplete, but it never assigns two servers with the same configuration.
\end{proof}

\begin{lemma}
\label{lemma:2}
If any two servers possess same configuration, they must be in different configuration clocks. 
\end{lemma}

\begin{proof}
We prove this lemma by contradiction. We claim that servers assigned with the same configuration have the same \texttt{confClock}, say $k$. Then, at configuration clock $k$, a leader must have assigned the same configuration to more than one server, which contradicts Lemma~\ref{lemma:1}. 
\end{proof}

Since \scap tolerates $f$ benign faults, from Lemma~\ref{lemma:2}, there are at most $f$ servers that are in different configuration clocks while possessing the same configuration. Suppose in clock $k$, a server $S_1$ has a configuration $\pi^{(\mathcal{P_*}, k)}$. In clock $k+1$, server $S_1$ crashes, and $\pi^{(\mathcal{P_*})}$ is assigned to $S_2$; thus, $S_2$'s configuration is $\pi^{(\mathcal{P_*}, k+1)}$. Then, $S_2$ crashes and $\pi^{(\mathcal{P_*})}$ is assigned to $S_3$. In this way, since \scap tolerates $f$ failures, $\pi^{(\mathcal{P_*})}$ can be iteratively assigned to $S_i$ in clock $k+i-1$ where $i \leq f$.

\begin{theorem}[Configuration uniqueness]
\label{theorem:config-uniqueness}
No two nonfaulty servers possess the same prioritized configuration.
\end{theorem}

\begin{proof}
Lemma~\ref{lemma:1} guarantees that, at any given time, candidates are prioritized by configurations. From Lemma~\ref{lemma:2}, configurations are prioritized by their logical clocks. Since configuration clocks grow monotonically along with physical time, the rearrangement of configurations always atomically scatters servers into different priorities.
\end{proof}

Due to the problem of competing candidates, Raft ensures liveness if the system ultimately elects a leader. That is, \emph{every nonfaulty server eventually decides a value.} However, with bounded messaging delay, Raft cannot guarantee a lower bound of the eventuality; competing candidates may often occur if timer timeouts are not ideally chosen. In contrast, equipped with SCA and PPF, \scap ensures that a new leader can be elected by at most $f+1$ leader elections in the worst case. If the highest-priority candidate does not fail, it will become the new leader and terminates leader election in one election campaign.

\begin{lemma}
\label{lemma:one-campaign}
\scap terminates a leader election phase in one campaign with nonfaulty candidates.
\end{lemma}

\begin{proof}
With Theorem~\ref{theorem:config-uniqueness}, in the phase of leader election, no two candidates possess the same prioritized configurations. Thus, at most one legit election campaign takes place in a given term. With nonfaulty candidates, the candidate that increases its term to the highest will be elected, which subsequently terminates the leader election.
\end{proof}

\begin{theorem}[Strong liveness]
Every nonfaulty server decides a value after at most $f+1$ leader elections.
\end{theorem}

\begin{proof}
\scap and Raft both tolerate $f$ non-Byzantine failures with a total of $2f+1$ servers. With a correct leader, nonfaulty servers decide a value in two rounds of heartbeats under normal operation. When the leader fails, normal operation resumes after a new leader is elected. In \scap, if $f$ failures occur on every candidate that has the highest-priority among nonfaulty ones, \scap must wait for $f+1$ election campaigns before a correct candidate appears. With Lemma~\ref{lemma:one-campaign}, \scap elects a new leader in one leader election campaign when candidates are nonfaulty. Therefore, normal operation resumes after at most $f+1$ leader elections.
\end{proof}

\begin{theorem}[Complexity]
\label{theorem:complexity}
\scap leader election has a message transmission complexity of $\mathcal{O}(n^2)$.
\end{theorem}

\begin{proof}
Upon leader failures, every server is able to start an election campaign after its timer expires and subsequently broadcasts a request to all the other servers. If a candidate is qualified for the next leadership, other servers send a vote to the candidate. Therefore, during the election, the message transmission complexity is $\mathcal{O}(n^2)$.
\end{proof}

\begin{figure*}[t]
\centering
\includegraphics[width=7.1in]{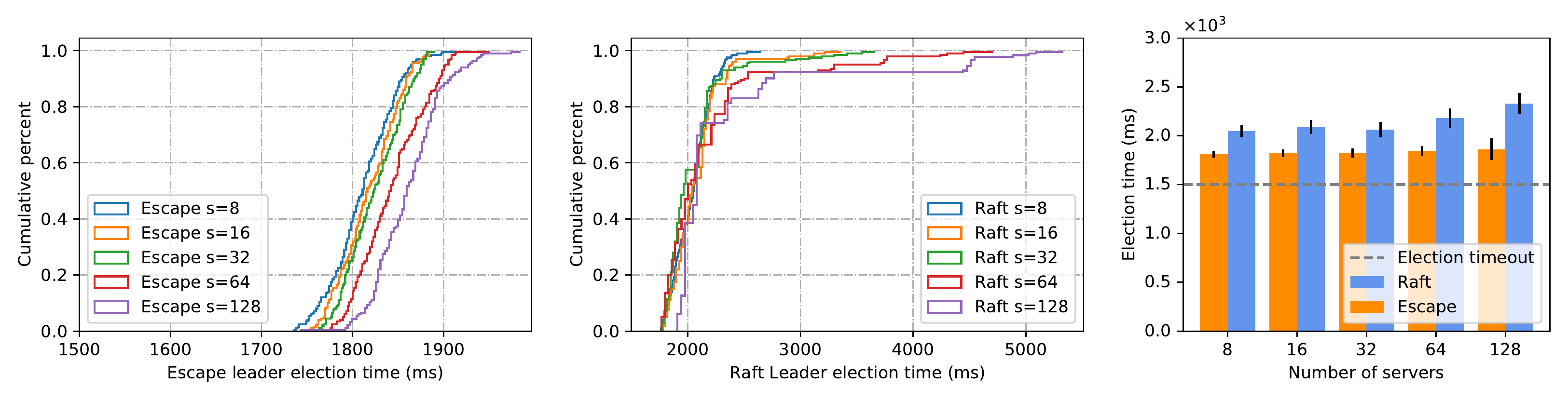}
\caption{Leader election (LE) time comparison of \scap and Raft at increasing scales. The first two figures indicate the cumulative percent of distribution of LE time for $1000$ runs of each scale, and the last figure compares the average LE time of both protocols.}
\label{fig:le_time_normal}
\end{figure*}

Additionally, the best case, where only one election campaign is initiated, has a message transmission complexity of $\mathcal{O}(n)$. Compared with Raft, \scap\ has a higher probability to achieve the best case. The highest priority and the shortest election timeout is assigned to the most up-to-date candidate, who presumably will firstly detect leader failure and become the next leader. 




\section{Evaluation}
\label{sec:evaluation}
In this section, we show experimental comparisons majorly between \scap and Raft. \scap aims to avoid split votes and reduce the leader election time, so the comparison was conducted for leader election under various scenarios including crash and omission failures.

\subsection{Experimental Setup}

We implemented \scap and Raft prototypes using the Golang programming language and deployed the prototypes on 4, 8, 16, 32, 64 and 128 VM instances on Compute Canada Cloud. Instances are located in the same data center; the raw network latency between two VMs is less than $2$ ms, and the network bandwidth is around $400$ Megabytes/second. We simulated a geo-distributed setup by using \texttt{NetEm} to implement a higher network latency that uniformly distributes from $100$ to $200$~ms. Each instance includes a machine with two 2.40 GHz Intel Core processors (Skylake) with a cache size of 16 MB, 8 GB of RAM, and 20 GB of disk space running on Ubuntu 18.04.1 LTS.

\subsection{Election Time under Leader Failures}
\label{sec:failure-free-le}
We first compared the leader election time for \scap and Raft under leader failures. In this case, we repeatedly crashed the leader of a cluster of $8$, $16$, $32$, $64$, and $128$ servers for $1000$ runs of leader election at each scale. In the Raft cluster, election timeouts were set to $1500$-$3000$~ms, which is the value range recommended by Raft for our network latency~\cite{ongaro2014search}. In \scap clusters, the $baseTime$ of election timeouts for servers was set to $1500$~ms with $k{=}500$ for Equation~\ref{eq:1}. 
In practice, to avoid simultaneous timeouts among servers, $k$ can be set $\times2$ higher than the network latency. Thus, \scap allows the ``potential leader'' to complete its election  before other servers trigger timeouts. This setting assists the occurrence of the best case in leader election (discussed after Theorem~\ref{theorem:complexity}). 
The leader election time is recorded including the detection of the leader crash and the election of a new leader. In each run, candidates need to collect votes from a majority of servers in their terms to claim leadership. 
For example, in an $8$-server cluster, the quorum size is $5$; after the leader crashes, a new leader is the candidate that has collected $4$ votes from the remaining servers (candidates vote only for themselves).

Figure~\ref{fig:le_time_normal} shows the evaluation result for \scap and Raft operating in clusters at increasing scales (denoted by $s$). The cumulative percent of the leader election time distribution for both protocols are presented in the first two figures. In \scap, all the election campaigns were completed within $2000$~ms, with no occurrence of split votes; in contrast, when $s \geq 32$, less than $40\%$ of Raft's election campaigns were completed within $2000$~ms. In the $128$-server cluster, more than $17\%$ of election campaigns experienced split votes, and their election time exceed $4500$~ms. The last figure shows the comparison of averaged election times. Compared with Raft, \scap shortens the leader election time by $11.6\%$ and $21.3\%$ at sizes of $8$ and $128$ servers, respectively.

\begin{figure*}[t]
\centering
\includegraphics[width=7.1in]{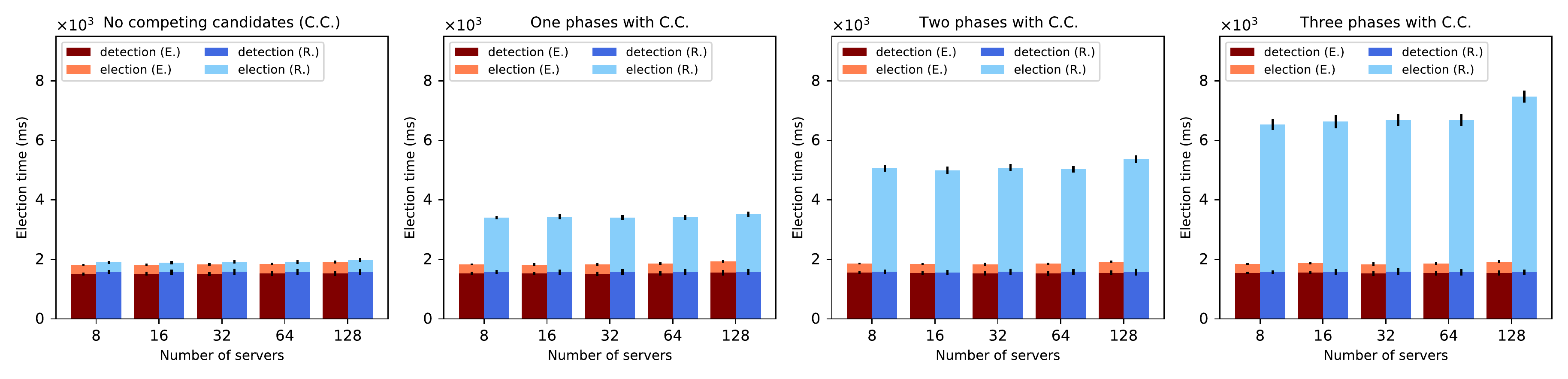}
\caption{Leader election time performance in zero, one, two, and three phases with competing candidates (C.C.) at five varying scales. The detection period is recorded between when a leader crashes and a candidate appears. The election period is recorded between when a candidate starts an election campaign and a new leader is elected. \scap benefits from its dynamic configuration assignments and performs leader election in a single campaign regardless of the number of competing candidates.}
\label{fig:three_phases_competing_candidates}
\end{figure*}

\subsection{Election Time in Multiple Phases}
From Section~\ref{sec:failure-free-le}, the evaluation result shows that \scap outperforms Raft in leader election under single leader failures. Because \scap does not encounter split votes in a given term, it reduces the leader election on average. In order to observe the performance for both protocols when configurations lead to multiple phases with competing candidates, we evaluated the leader election time where five clusters of various scales elect a new leader in zero, one, two, and three phases with competing candidates (shown in Figure~\ref{fig:three_phases_competing_candidates}).

Raft witnesses a surge in leader election time when split votes repeat. Both protocols exhibit similar performance without competing candidates: election converges in $1812$~ms (\scap $s{=}8$) to $1976$~ms (Raft $s{=}128$) for clusters at different scales. However, Raft suffers from a provisional livelock when competing candidates emerge. Although Raft and \scap take similar time to detect leader failures, competing candidates impede Raft from completing the current election. The livelock lasts for approximately the number of phases with competing candidates $\times$ the election timeout. With three phases of competing candidates, Raft takes approximately $6535$~ms to elect a leader in a cluster of $8$ servers, and $7473$~ms in a cluster of $128$ servers, nearly four times as long as \scap ($1924$~ms at $s{=}128$).
    
On the contrary, \scap completes all election campaigns within $2000$~ms throughout different evaluations, regardless of the number of phases with competing candidates. \scap no longer elects a leader based on the success that a candidate collects votes faster than its opponents and forms a majority. The highest-termed candidate finally supersedes the other election campaigns and becomes the new leader. In a $128$-server cluster, compared with Raft, \scap reduces the election time by $44.9\%$, $64.2\%$, and $74.3\%$ under one, two, and three phases with competing candidates, respectively.

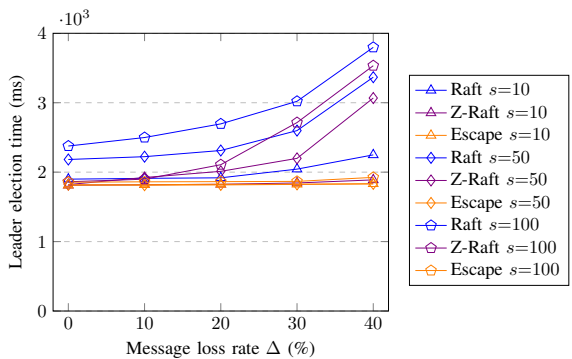
\begin{figure}[b!]
    \centering
    \begin{adjustbox}{width=0.85\linewidth}
    \begin{tikzpicture}
    \pgfplotsset{width=8cm, compat=1.9, scaled y ticks=base 10:-3,}
    
    \begin{axis}[
        title={},
        xlabel={Message loss rate $\Delta$ (\%)},
        ylabel={Leader election time (ms)},
        xmin=-2, xmax=42,
        ymin=0, ymax=4000,
        xtick={0, 10, 20, 30, 40},
        ylabel near ticks,
        legend pos=south west,
        legend columns=1, 
        legend style={font=\normalsize},
        ymajorgrids=true,
        grid style=dashed,
        mark size=3pt,
        legend cell align={left},
        legend style={at={(1.3, 0.85)}, anchor=north},
    ]
    
    \definecolor{patriarch}{rgb}{0.5, 0.0, 0.5}
    
    \addplot[color=blue, 
    mark=triangle,
    ]
    table{data/raft10mloss.data};
    \addlegendentry{Raft $s{=}10$}
    
    \addplot[color=patriarch, 
    mark=triangle, 
    ]
    table{data/z10mloss.data};
    \addlegendentry{Z-Raft $s{=}10$}
    
    \addplot[color=orange,
    mark=triangle, 
    ]
    table{data/escape10mloss.data};
    \addlegendentry{Escape $s{=}10$}

    \addplot[color=blue, 
    mark=diamond,
    ]
    table{data/raft50mloss.data};
    \addlegendentry{Raft $s{=}50$}
    
    \addplot[color=patriarch, 
    mark=diamond, 
    ]
    table{data/z50mloss.data};
    \addlegendentry{Z-Raft $s{=}50$}
    
    \addplot[color=orange,
    mark=diamond, 
    ]
    table{data/escape50mloss.data};
    \addlegendentry{Escape $s{=}50$}

    \addplot[color=blue, 
    mark=pentagon, 
    ]
    table{data/raft100mloss.data};
    \addlegendentry{Raft $s{=}100$}
    
    \addplot[color=patriarch, 
    mark=pentagon, 
    ]
    table{data/z100mloss.data};
    \addlegendentry{Z-Raft $s{=}100$}
    
    \addplot[color=orange,
    mark=pentagon, 
    ]
    table{data/escape100mloss.data};
    \addlegendentry{Escape $s{=}100$}
    
    \end{axis}
    \end{tikzpicture}
    \end{adjustbox}
    
    \caption{Leader election comparisons under message loss. 
    When $\Delta{=}40\%$, in each broadcast, $40\%$ of the servers were not able to receive messages sent by a leader or candidates.}
    \label{fig:le_message_loss}
\end{figure}

\subsection{Election Time under Message Loss}
In addition to multiple phases with competing candidates, we evaluated the leader election performance under message loss, a scenario that frequently occurs in practice. \scap and Raft depend on a majority of servers participating in the consensus, so we chose to evaluate leader election under message loss rates~($\Delta$) of $0\%$, $20\%$, $30\%$, $40\%$. At each rate, a broadcast only reaches $1{-}\Delta$ servers. For example, in a cluster of $10$ servers and $\Delta{=}20\%$, a sender (leader or candidate) randomly omits two servers in each broadcast, whereas if $\Delta{=}0\%$, no message loss materializes.

Zookeeper~\cite{hunt2010zookeeper} implemented a leader election mechanism~\cite{zookeeperleaderelection} using unique server IDs to set priorities, which is similar to \scap's SCA method without PPF. However, if applying this approach to Raft, the priority assignment is not atomic. We applied Zookeeper's leader election approach in Raft and refer to it as Z-Raft. 
To compare the performance of \scap, Raft, and Z-Raft at various scales under varying rates of message loss, we chose to perform the experiment in clusters of $10$, $50$, and $100$ servers. 

Figure~\ref{fig:le_message_loss} shows the evaluation results averaged over $1000$ runs. Without message loss ($\Delta{=}0\%$), the election time of Raft is slightly higher than that of \scap and Z-Raft. However, message loss severely exacerbates split votes in Raft, especially at large scales. When $\Delta{=}40\%$ in the Raft cluster, $40\%$ of servers become unqualified candidates since they do not have up-to-date logs. Election campaigns triggered by unqualified candidates are in the same term and stranded in their election campaigns until a qualified candidate with a higher term wins the competition.

The results of Z-Raft and \scap exhibit similarity when the rate of message loss is relatively small. Since Z-Raft does not atomically arrange configurations, a low rate of message loss does not significantly jeopardize the effectiveness of fixed priorities. However, when message loss rates increase, without dynamically rearranging priorities, servers that are initially assigned with higher priorities are more likely to become stale servers whose logs are no longer up-to-date. This situation defeats the purpose of using priorities to rank candidates. Thus, when $\Delta$ increases, the gap of election times among Raft, Z-Raft, and \scap expands.

In contrast, \scap assigns the ``best'' configuration (the highest priority and shortest election timeout) to the most up-to-date server through SCA and PPF. With message loss in the network, if the most up-to-date server changes, the ``best'' configuration is dynamically rearranged to the qualified server. Because the configuration rearrangement is atomic, candidates are differentiated in real time based on log responsiveness through periodic heartbeats. When the leader crashes, the server firstly triggers an election campaign and completes the campaign in one phase (all the other servers will vote for this server), escaping from potential competitions.

The evaluation results show the quantitative comparison among three protocols. In a cluster size of $10$ servers, compared with Raft, when $\Delta{=}10\%$ and $\Delta{=}40\%$, Z-Raft reduces election time by $9.8\%$ and $14.3\%$, respectively, while \scap reduces election time by $9.6\%$ and $19\%$, respectively. The reduction in election time is more significant at large scales for \scap, especially at a high rate of message loss. In a cluster of $100$ servers, \scap reduces the leader election time by $21.4\%$ and $49.3\%$, when $\Delta{=}10\%$ and $\Delta{=}40\%$, respectively. 


\section{Related Work}
\label{sec:related-work}

Consensus algorithms coordinate servers to agree on proposed values in distributed systems. 
To efficiently achieve agreement among servers, numerous consensus algorithms use a designated leader to coordinate the process for committing values. 
These algorithms are called  leader-based consensus algorithms and contain election protocols to choose a unique server to play a particular role.

With an assumption of no server failures and an asynchronous network, the Chang and Roberts algorithm arranges an array of servers as a logical ring~\cite{changandroberts}, where a server communicates with its descendant in a clock-wise manner, and a single process that has the largest identifier is elected as the leader. In addition, the Bully algorithm~\cite{garcia1982elections} addresses fault tolerance by relying on communication among servers; it assumes a synchronous network and uses timeouts to detect server failures. The server with the highest server ID will be elected as the leader.

More recently, the celebrated Paxos algorithm~\cite{lamport1998part, lamport2001paxos} implicitly utilize a global leader for reaching consensus; every server that assumes a proposer role is allowed to broadcast proposals to be accepted. Nonetheless, the concept of leadership hides behind the consensus process that refuses values carrying a smaller proposal number~\cite{ongaro2014search}. To avoid conflicts and optimize performance, approaches using explicit leader roles have been developed~\cite{chandra2007paxos, burrows2006chubby, lamport2005generalized, lamport2006fast}. 
In addition, Viewstamped Replication (VR)~\cite{oki1988viewstamped, liskov2012viewstamped} provides an alternative way to achieve consensus using a leader-based approach. Although Raft shares similarities with Paxos and VR, Raft uses strong leadership where the messages flow only from the leader to other servers; all other servers passively synchronize states from the leader, resolving conflicts by obeying the leader's commands.
Utilizing strong leadership reduces the types of messages and thus improves Raft's understandability, making Raft widely deployed in practical large-scale distributed systems such as Baidu File System~\cite{baiduRaft}, SDN designs~\cite{sakic2017response}, and HyperLedger Kafka~\cite{androulaki2018hyperledger}. Thus, the leader election mechanism is crucial because a Raft system fails without a leader. However, Raft's leader election may result in competing election campaigns that needlessly prolong the election time. The performance evaluation of Raft in~\cite{huang2019performance} hinted at this problem, especially when the network splits and drops message. 


Although Zookeeper~\cite{hunt2010zookeeper} ushered a way to prioritize servers using server IDs as identifiers~\cite{zookeeperleaderelection}, this approach does not atomically rearrange priorities based on server log responsiveness. Consequently, the highest-priority server is guaranteed as the most up-to-date. In this case, highest-priority servers can be defeated in leader election campaigns and can not be elected as the new leader, especially under problematic networks (e.g., packet loss). In contrast, \scap does not only use server identifiers to assign priorities, but also periodically rearranges priorities and election timeouts in the cluster.
Through the proposed election mechanism, \scap prepares and manages a pool of prioritized candidates, resolving potential conflicts before the next election takes place. 
Furthermore, \scap shares Raft's safety and liveness properties and maintains understandability and modularity as well.

\section{Conclusions}
In this paper, we develop a new consensus protocol called \scap that enables fast leader election. \scap assigns stochastic configurations to prioritize servers. Higher-priority configurations are dynamically distributed to more up-to-date servers via the probing patrol function, forming a pool of prioritized candidates. 
The top candidate, with the highest priority configuration, in the pool is groomed as a ``future leader". Thus, leader election completes without suffering from split votes, taking precautions against leader failures before the next election takes place. 
\scap can also be adapted in other election protocols to resolve potential competition in advance. The evaluation results show improvements of leader election time compared with Raft and ZooKeeper under leader failures, especially under message loss.

\Urlmuskip=0mu plus 1mu\relax
\bibliographystyle{IEEEtran}
\bibliography{ref.bib}

\begin{thebibliography}{10}
\providecommand{\url}[1]{#1}
\csname url@samestyle\endcsname
\providecommand{\newblock}{\relax}
\providecommand{\bibinfo}[2]{#2}
\providecommand{\BIBentrySTDinterwordspacing}{\spaceskip=0pt\relax}
\providecommand{\BIBentryALTinterwordstretchfactor}{4}
\providecommand{\BIBentryALTinterwordspacing}{\spaceskip=\fontdimen2\font plus
\BIBentryALTinterwordstretchfactor\fontdimen3\font minus
  \fontdimen4\font\relax}
\providecommand{\BIBforeignlanguage}[2]{{%
\expandafter\ifx\csname l@#1\endcsname\relax
\typeout{** WARNING: IEEEtran.bst: No hyphenation pattern has been}%
\typeout{** loaded for the language `#1'. Using the pattern for}%
\typeout{** the default language instead.}%
\else
\language=\csname l@#1\endcsname
\fi
#2}}
\providecommand{\BIBdecl}{\relax}
\BIBdecl

\bibitem{shvachko2010hadoop}
K.~Shvachko, H.~Kuang, S.~Radia, R.~Chansler \emph{et~al.}, ``The hadoop
  distributed file system.'' in \emph{MSST}, vol.~10, 2010, pp. 1--10.

\bibitem{ousterhout2010case}
J.~Ousterhout, P.~Agrawal, D.~Erickson, C.~Kozyrakis, J.~Leverich,
  D.~Mazi{\`e}res, S.~Mitra, A.~Narayanan, G.~Parulkar, M.~Rosenblum
  \emph{et~al.}, ``The case for {RAMClouds}: scalable high-performance storage
  entirely in dram,'' \emph{ACM SIGOPS Operating Systems Review}, vol.~43,
  no.~4, pp. 92--105, 2010.

\bibitem{burrows2006chubby}
M.~Burrows, ``The {Chubby} lock service for loosely-coupled distributed
  systems,'' in \emph{Proceedings of the 7th symposium on Operating systems
  design and implementation}.\hskip 1em plus 0.5em minus 0.4em\relax USENIX
  Association, 2006, pp. 335--350.

\bibitem{hunt2010zookeeper}
P.~Hunt, M.~Konar, F.~P. Junqueira, and B.~Reed, ``Zookeeper: Wait-free
  coordination for internet-scale systems.'' in \emph{USENIX annual technical
  conference}, vol.~8, no.~9.\hskip 1em plus 0.5em minus 0.4em\relax Boston,
  MA, USA, 2010.

\bibitem{lamport1998part}
L.~Lamport, ``The part-time parliament,'' \emph{ACM Transactions on Computer
  Systems (TOCS)}, vol.~16, no.~2, pp. 133--169, 1998.

\bibitem{oki1988viewstamped}
B.~M. Oki and B.~H. Liskov, ``Viewstamped replication: A new primary copy
  method to support highly-available distributed systems,'' in
  \emph{Proceedings of the seventh annual ACM Symposium on Principles of
  distributed computing}.\hskip 1em plus 0.5em minus 0.4em\relax ACM, 1988, pp.
  8--17.

\bibitem{ongaro2014search}
D.~Ongaro and J.~Ousterhout, ``In search of an understandable consensus
  algorithm,'' in \emph{2014 USENIX Annual Technical Conference (USENIX ATC
  14)}, 2014, pp. 305--319.

\bibitem{sahoo2004failure}
R.~K. Sahoo, M.~S. Squillante, A.~Sivasubramaniam, and Y.~Zhang, ``Failure data
  analysis of a large-scale heterogeneous server environment,'' in
  \emph{International Conference on Dependable Systems and Networks,
  2004}.\hskip 1em plus 0.5em minus 0.4em\relax IEEE, 2004, pp. 772--781.

\bibitem{servercrash}
``Frequency of server failure based on the age of the server,''
  \url{https://www.statista.com/statistics/430769/annual-failure-rates-of-servers/}.

\bibitem{hardwarefails}
L.~StorageCraft~Technology, ``Which hardware fails the most and why,''
  \url{https://blog.storagecraft.com/hardware-failure/}, accessed: 2020.

\bibitem{bowers2011tell}
K.~D. Bowers, M.~Van~Dijk, A.~Juels, A.~Oprea, and R.~L. Rivest, ``How to tell
  if your cloud files are vulnerable to drive crashes,'' in \emph{Proceedings
  of the 18th ACM conference on Computer and communications security}, 2011,
  pp. 501--514.

\bibitem{gill2011understanding}
P.~Gill, N.~Jain, and N.~Nagappan, ``Understanding network failures in data
  centers: measurement, analysis, and implications,'' in \emph{Proceedings of
  the ACM SIGCOMM 2011 Conference}, 2011, pp. 350--361.

\bibitem{redisLeaderElection}
Redis, ``Redis cluster specification,''
  \url{https://redis.io/topics/cluster-spec}.

\bibitem{googleServicesOutages}
Wikipedia, ``2020 google services outages,''
  \url{https://en.wikipedia.org/wiki/2020\_Google\_services\_outages},
  accessed: 2021.

\bibitem{amazonServicesOutages}
Verge, ``Amazon web services outages,''
  \url{https://www.theverge.com/2020/11/25/21719396/amazon-web-services-aws-outage-down-internet}.

\bibitem{baiduRaft}
``Baidu file system,'' https://github.com/baidu/bfs, accessed: 2016-03-1.

\bibitem{sakic2017response}
E.~Sakic and W.~Kellerer, ``Response time and availability study of {RAFT}
  consensus in distributed {SDN} control plane,'' \emph{IEEE Transactions on
  Network and Service Management}, vol.~15, no.~1, pp. 304--318, 2017.

\bibitem{androulaki2018hyperledger}
E.~Androulaki, A.~Barger, V.~Bortnikov, C.~Cachin, K.~Christidis, A.~De~Caro,
  D.~Enyeart, C.~Ferris, G.~Laventman, Y.~Manevich \emph{et~al.}, ``Hyperledger
  fabric: a distributed operating system for permissioned blockchains,'' in
  \emph{Proceedings of the Thirteenth EuroSys Conference}.\hskip 1em plus 0.5em
  minus 0.4em\relax ACM, 2018, p.~30.

\bibitem{brown2016corda}
R.~G. Brown, J.~Carlyle, I.~Grigg, and M.~Hearn, ``Corda: an introduction,''
  \emph{R3 CEV, August}, vol.~1, p.~15, 2016.

\bibitem{zhang2021prosecutor}
G.~Zhang and H.-A. Jacobsen, ``Prosecutor: an efficient bft consensus algorithm
  with behavior-aware penalization against byzantine attacks,'' in
  \emph{Proceedings of the 22nd International Middleware Conference}, 2021, pp.
  52--63.

\bibitem{azurele}
M.~Azure, ``Azure leader election,''
  \url{https://docs.microsoft.com/en-us/azure/architecture/patterns/leader-election}.

\bibitem{huang2019performance}
D.~Huang, X.~Ma, and S.~Zhang, ``Performance analysis of the {Raft} consensus
  algorithm for private blockchains,'' \emph{IEEE Transactions on Systems, Man,
  and Cybernetics: Systems}, 2019.

\bibitem{zookeeperleaderelection}
``{Fast Leader Election of ZooKeeper},''
  \url{https://github.com/apache/zookeeper/blob/master/zookeeper-server/src/main/java/org/apache/zookeeper/server/quorum/FastLeaderElection.java}.

\bibitem{changandroberts}
\BIBentryALTinterwordspacing
E.~Chang and R.~Roberts, ``An improved algorithm for decentralized
  extrema-finding in circular configurations of processes,'' \emph{Commun.
  ACM}, vol.~22, no.~5, p. 281–283, May 1979. [Online]. Available:
  \url{https://doi.org/10.1145/359104.359108}
\BIBentrySTDinterwordspacing

\bibitem{garcia1982elections}
H.~Garcia-Molina, ``Elections in a distributed computing system,'' \emph{IEEE
  Transactions on Computers}, no.~1, pp. 48--59, 1982.

\bibitem{lamport2001paxos}
L.~Lamport, ``Paxos made simple,'' \emph{ACM Sigact News}, vol.~32, no.~4, pp.
  18--25, 2001.

\bibitem{chandra2007paxos}
T.~D. Chandra, R.~Griesemer, and J.~Redstone, ``Paxos made live: an engineering
  perspective,'' in \emph{Proceedings of the twenty-sixth annual ACM symposium
  on Principles of distributed computing}.\hskip 1em plus 0.5em minus
  0.4em\relax ACM, 2007, pp. 398--407.

\bibitem{lamport2005generalized}
L.~Lamport, ``Generalized consensus and paxos,'' \emph{Microsoft Research},
  2005.

\bibitem{lamport2006fast}
------, ``Fast paxos,'' \emph{Distributed Computing}, vol.~19, no.~2, pp.
  79--103, 2006.

\bibitem{liskov2012viewstamped}
B.~Liskov and J.~Cowling, ``Viewstamped replication revisited,'' 2012.

\end{thebibliography}
\end{document}